\documentclass[]{lmcs}

\makeatletter
\def\endfront@text{}
\makeatother

\usepackage{amsmath}
\usepackage{graphicx}
\usepackage{moreverb}
\usepackage{subfig}

\usepackage{natbib}
\let\citeN\citet

\usepackage{hyperref}

\newtheorem{theorem}{Theorem}[section]

\newlength{\figwidth}
\multiply\figwidth by9
\divide\figwidth by10
\newlength{\halffigwidth}
\divide\halffigwidth by2

\def\setseparator{\mid}
\newcommand{\set}[2][\relax]{
  \ifx#1\relax
    \{#2\}
  \else
    \ifx#1\left
      \left\{#2\right\}
    \else
      \csname #1l\endcsname\{#2\csname #1r\endcsname\}
    \fi
  \fi
}
\newcommand{\setst}[3][\relax]{\set[#1]{#2\setseparator#3}}

\def\killprob{p_{\mathrm{kill}}}
\def\permitprob{p_{\mathrm{permit}}}

\def\gcr{\mathit{gcr}}
\def\scr{\mathit{scr}}

\begin{document}

\title{Effectiveness and detection of \\ denial of service attacks in Tor} 

\def\wesaddr{Department of Mathematics and Computer Science \\
Wesleyan University \\
Middletown, CT 06459 USA
}
\author{Norman Danner}
\address{\wesaddr}
\email{ndanner@wesleyan.edu}

\author{Sam DeFabbia-Kane}
\address{\wesaddr}
\email{sdefabbiakan@wesleyan.edu}

\author{Danny Krizanc}
\address{\wesaddr}
\email{dkrizanc@wesleyan.edu}

\author{Marc Liberatore}
\address{Department of Computer Science \\
University of Massachusetts Amherst \\
Amherst, MA 01003 USA}
\email{liberato@cs.umass.edu}

\begin{abstract}
Tor is one of the more popular systems for 
anonymizing near real-time communications on the Internet. 
Borisov et al.\ proposed a denial of service 
based attack on Tor (and related systems) that significantly
increases the probability of compromising the anonymity provided. 
In this paper, we analyze the effectiveness of the attack using both
an analytic model and simulation.  We also describe two algorithms for 
detecting such attacks, one deterministic and proved correct, the
other probabilistic and verified in simulation.
\end{abstract}

\titlecomment{Copyright ACM, 2013. This is the author's version of the work. It is posted here by permission of ACM for your personal use. Not for redistribution. The definitive version was published in \emph{Transactions on Information and System Security}, Volume 11 (2013), \url{http://dx.doi.org/10.1145/2382448.2382449}.}

\maketitle

\section{Introduction}

A low-latency anonymous communication system attempts to allow
near-real-time communication between hosts while hiding the identity
of these hosts from various types of observers (including each
other).  Such a system is
useful whenever communication privacy is desirable --- personal,
medical, legal, governmental, or financial applications all may
require some degree of privacy.  
Dingledine, Mathewson, and Syverson~\citeyearpar{tor-design} developed
the Onion Routing network \emph{Tor} for such
communication.  Tor anonymizes
communication by sending it along paths of anonymizing proxies,
encrypting messages in layers so that each proxy only knows its
neighbors in the path.

Syverson et al.~\citeyearpar{syverson-pet00} showed that such systems are
vulnerable to a passive adversary (one who does not modify traffic
in any way) who controls the first and last
proxies along such a path; roughly speaking, the attack involves a
cross-correlation of timing data.
Active attacks, and in particular, denial of service (DoS) attacks, can
increase the power of an otherwise limited attacker.  For example,
Dingledine, Shmatikov, and Syverson~\citeyearpar{sync-batching} analyzed the impact of DoS on different
configurations of mix networks.  The Crowds design paper~\citet{crowds:tissec}
examined the impact of circuit interruptions on anonymity.  And the
original Tor design paper describes various denial of service attacks.
More recently, Borisov et
al.~\citeyearpar{ccs07-doa} showed that an adversary willing to engage
in denial of service (DoS) could increase her probability of compromising
anonymity.  When a path is reconstructed after a denial of service,
new proxies are chosen, and thus the adversary has another chance to
be on the endpoints of the path.  

In this paper we analyze the denial-of-service attack in detail and
propose two detection algorithms.\footnote{We reported on preliminary
results in \cite{ddos-fc09}.}  
In Section~\ref{sec:borisov-attack}
we give a careful description of the attack in terms of a number
of parameters that the attacker might vary to avoid detection
(our model includes Borisov et al.'s attacker and the passive attacker
as special cases).  In Section~\ref{sec:effectiveness}
we assess the effectiveness of the
attacker as a function of these parameters.  We compare our analytic
results to a simulation of the attacker based on replaying data collected
from the deployed Tor network.
In Section~\ref{sec:detection-alg} we prove that an adversary engaging
in the DoS attack in an idealized Tor-like system can be detected by
probing at most $3n$ paths in the system, where $n$ is the number of proxies
in the system.
We give a more practical algorithm in
Section~\ref{sec:implementation}, implement it in simulation,
and show that it detects DoS attackers with low error rate.
In Section~\ref{sec:variants} we
discuss attackers that do not fit our model perfectly and show how
our detection algorithms might cope with such attackers.
In Section~\ref{sec:deployment} we discuss issues related to
deploying the detection algorithms we describe.
Finally, attacking and defending anonymity networks is an arms race;
in Section~\ref{sec:related} we discuss other attempts to detect and
defend against various kinds of attacks.  In particular, we
compare our more practical detection algorithm to the
``client-level'' algorithm for avoiding compromised tunnels
described by \citeN{das-borisov:securing-tor-tunnels} in that section.

\section{The Denial of Service Attack}
\label{sec:borisov-attack}

We model the Tor network with a fully connected undirected graph.%
\footnote{Some
individual Tor nodes may disable connections on specific ports or to
specific IP addresses.  We have not determined if these
significantly limit the graph.}
The vertices of the graph represent the Tor nodes (or relays), 
and the edges represent
network connections between nodes. 
We define $n$ to be the number of vertices.
For a DoS attack, we assume that the attacker controls some subset of
the relays; we may also use the term \emph{compromised} to describe
such relays.

A Tor client creates 
\emph{circuits} (also referred to as paths or tunnels) consisting of three
nodes;
in our model, this equates to  a path containing 
three vertices (in order) and the corresponding edges between them.  
The first node is referred to as the \emph{entry} node and the
last as the \emph{exit} node.
Application level communications between an initiator and a responder
are then passed through the circuit.  We
assume that 
if the adversary controls the entry and exit nodes on a circuit,
then she can in fact determine whether or not the traffic passing through the
entry node is the same as the traffic passing through the exit node
(and hence she can match the initiator with the recipient of the traffic).
An early version of such an attack is given by
\citeN{timing-fc2004} and a more sophisticated version 
by \citeN{murdoch-zielinski:pet2007}.
A circuit is \emph{compromised} if 
at least one node on the circuit is compromised
and the circuit is
\emph{controlled} if both the entry and exit nodes are compromised.

Syverson et al.~\citeyearpar{syverson-pet00} observe that 
if all nodes may act as exit nodes, then
a passive adversary controls a circuit with probability
$\frac{c^2}{n^2}$, where $c$ is the
number of nodes controlled by the attacker.  Since controlling
middle nodes is of little use, we might also consider an attacker
who selectively compromises exit nodes.  In this case
the probability of control is
$\frac{c^2}{nc'}$, where $c'$ is the number of compromised exit nodes.
Levine et al.~\citeyearpar{timing-fc2004} observe that if long-lived connections
between an initiator and responder are reset at a reasonable rate then
such an attack will be able to compromise anonymity with high probability
within $O(\frac{n^2}{c^2} \ln n)$ resets.

But Tor's node-selection algorithm is more sophisticated than
portrayed in these models.\footnote{This paper
refers to the algorithm and implementation of version 0.2.1.27 of Tor.}
Nodes are assigned flags by directory servers, and among these
flags are ``Guard'' and ``Exit.''  Entry and exit nodes will only be chosen
from nodes that are flagged Guard and Exit, respectively, and
a given guard, middle, or exit node is chosen with probability proportional to
its contribution to the total guard, network, or exit bandwidth.%
\footnote{Tor imposes a maximum ``believable'' bandwidth on any relay
when computing these probabilities; at the time we collected our
trace data that
we describe in later sections, this cap was $10$~MB/s, and we impose
this same cap on the bandwidths we recorded.}
Furthermore, unless the total bandwidth
contributed by exit nodes is greater than $1/3$ of the total network
bandwidth, exit nodes will not be chosen for any other position on a
circuit.  If the contribution is
$t > 1/3$, then the probability with which an exit node will
be chosen for a non-exit position is weighted by~$t-1/3$, so $t$ must
be significantly greater than $1/3$ for an exit node to be chosen in a
non-exit position with any non-trivial probability.  The same applies
to guard nodes.

Even further, guard nodes are not chosen each time a circuit is built
by a client.  Instead, a client constructs a list
of \emph{guard nodes} (currently of length $3$) 
chosen by the above algorithm when first started,
and thereafter always choses entry nodes uniformly at random
from that list when constructing circuits (guard nodes were
first described by \citeN{wright-sosp2003}).
New nodes are added to the list only if there are fewer than three
reachable nodes in the list, and a node is removed from the list only
if it has not been reachable for some time.
This protocol is intended to reduce the likelihood that a 
client eventually constructs a circuit that is controlled by an
attacker.  If new entry nodes were chosen for every circuit, then the
probability that a compromised entry node is chosen
converges to~$1.0$.  
The guard node list changes the calculus a bit:  the probability
of choosing a compromised node to put into the list is (hopefully)
low, and
if there are no such nodes in the list, then the client
will never be subject to a traffic confirmation attack of the sort
referenced earlier.
However, if there is a compromised node in the list, then
any circuit created by the client will have a compromised
entry node with fairly high probability.

In order to further improve the chances of controlling
a circuit, a number of researchers 
\cite{syverson-ieee06,low-resource-weps07,ccs07-doa}
have suggested that compromised nodes that occur on paths in which they
are not the first or last node artificially create a reset event 
by dropping the connection.  
Borisov et al.~\citeyearpar{ccs07-doa} analyze the following version
of this attack on Tor.  If the attacker controls one of the relays
of a given circuit, she can use various timing techniques to determine
if she controls both endpoints.  If she does not, she kills the circuit,
forcing the client to build another.
In the terminology we introduce below, 
Borisov et~al.'s attacker can be described
as killing circuits that are compromised but not controlled in an
effort to increase the number of circuits that she controls.

Killing a large number of circuits may make an attacker's nodes stand
out from other nodes on the network, so the attacker may try to
``fit in'' by not always killing
compromised but uncontrolled circuits.  In the notation introduced
below, the attacker may have $\killprob < 1.0$.  At the same time,
nodes on the Tor network do fail, for example because they have reached
bandwidth limits or because of network failures.  So an attacker might
also occasionally kill circuits that she controls in an attempt to
look ``more realistic.''  In the notation below, the attacker
may have $\permitprob < 1.0$.

So some relevant parameters for a denial-of-service attacker are 
the bandwidth contributed by attacker nodes (as this determines the
probability with which her nodes are chosen as part of a circuit)
and the probabilities of killing compromised uncontrolled circuits and 
permitting controlled circuits.
More specifically, we identify the following partial list of
parameters necessary to model the attacker:
\begin{enumerate}
\item $g$, the ratio of compromised guard bandwidth to total
guard bandwidth.
\item $e$, the ratio of compromised exit bandwidth to total
exit bandwidth.
\item $\killprob$, the probability that the attacker kills
a compromised but uncontrolled circuit.
\item $\permitprob$, the probability that the attacker permits
a circuit that she controls to be used.
\end{enumerate}
We make the following assumptions about the attacker:
\begin{enumerate}
\item $g>0$, $e>0$, and $\permitprob > 0$ (otherwise
the attacker can never control a circuit, or always kills any
circuit she controls).
\item The attacker only controls nodes with the Guard or
Exit flags (or both) set.
\item The attacker is local---i.e., the attacker can observe and
modify traffic passing through nodes she controls, but not other
traffic.
\end{enumerate}
So the attacker described by Borisov et al. has
$(\killprob, \permitprob) = (1.0, 1.0)$ and the passive
attacker who never kills any circuits
has $(\killprob, \permitprob) = (0.0, 1.0)$.
\begin{table}
\begin{tabular}{p{.25\textwidth}p{.65\textwidth}}
\multicolumn{2}{l}{\textbf{Parameters describing the attacker}} \\
$g$/$e$/$z$			& The ratio of compromised guard/exit/guard-exit
					  bandwidth to total guard bandwidth. \\
$\killprob$			& The probability that the attacker kills a compromised
					  but uncontrolled circuit. \\
$\permitprob$		& The probability that the attacker permits (does
					  not kill) a controlled circuit. \\
\multicolumn{2}{l}{\textbf{Parameters describing the network}} \\
$n$					& The number of relays in the Tor network. \\
$\gamma$/$\eta$/$\zeta$		& The ratio of guard/exit/guard-exit
					  bandwidth to total bandwidth. \\
\multicolumn{2}{l}{\textbf{Other quantities and terms}} \\
$K$					& The number of circuit-creation attempts a client makes 
					  before giving up. \\
\textit{Compromised circuit} 	& The attacker controls at least one relay
								  on the circuit. \\
\textit{Controlled circuit}		& The attacker controls the entry and exit
								  relays of the circuit.
\end{tabular}
\caption{Notation and terminology used in Section~\ref{sec:borisov-attack}.}
\end{table}

We also make the following assumptions about the Tor network:
\begin{enumerate}
\item The only reason for relay failure is a compromised relay
killing a circuit.
\item Paths are chosen with replacement.
\end{enumerate}
Of course, the first assumption is unreasonable, and the second
is false:  not only are paths chosen without replacement,
but in fact no two nodes on a path can belong to the same /16 subnetwork.
We make these assumptions about the network because modeling the
actual behavior in our analytic model of the attacker is extremely
difficult.  We assess the impact of these two assumptions on our
analytic model in Section~\ref{sec:sim-results}.

Tor's modified bandwidth-weighting algorithm 
chooses each relay with probability proportional to its weighted
bandwidth, where the weight assigned to a relay depends on
the position being selected for (guard, middle, or exit) and
the flags of the relay.  To fully describe the attacker,
we need to compute these probabilities.
A \emph{guard-only} relay is one with the 
Guard flag set and the Exit flag not set; this is
in distinction to a \emph{guard} relay, which is one
in which the Guard flag is set, irrespective of the
Exit flag.  We define \emph{exit-only} and \emph{exit} relays
similarly.
A \emph{guard-exit} relay is one with both flags set.
The bandwidth weights are defined in terms of the following values:
\begin{itemize}
\item $G$, $E$, and $T$, the guard-, exit-, and total bandwidth
of the network, respectively;
\item $\gamma = G/T$; $\eta = E/T$.
\item $w_{G_0} = 1-1/3\gamma$; $w_{E_0} = 1-1/3\eta$; and
$w_Z = w_{G_0}w_{E_0} = (1-1/3\gamma)(1-1/3\eta)$.
\end{itemize}
The various weights are given in Table~\ref{tbl:weights}.
\begin{table}
\begin{center}
\begin{tabular}{cc|cccc}
& & \multicolumn{4}{c}{Flags} \\
& & Guard-only & Exit-only & Guard-exit & None \\
\hline
& Guard	& 1.0	& 0.0	& $w_{E_0}$		& 0.0 \\
Position & Middle	& $w_{G_0}$	& $w_{E_0}$	& $w_Z$	& 1.0 \\
& Exit	& 0.0	& 1.0	& $w_{G_0}$		& 0.0
\end{tabular}
\end{center}
\caption{Bandwidth weights assigned by Tor based on position and
flags.}
\label{tbl:weights}
\end{table}

A node $R$ is chosen for a given position with probability
$b'/T'$, where $b'$ is the weighted bandwidth of $R$ and
$T'$ is the total bandwidth of all nodes weighted for that position.
So to compute the probability that a compromised node is chosen in any
position, we have to consider the possible combination of
flags of that node.
To that end, define
\begin{itemize}
\item $G_0$, $E_0$, and $Z$ to be the guard-only, exit-only, and
guard-exit bandwidth, respectively;
\item $\gamma_0 = G_0/T$; $\eta_0 = E_0/T$; $\zeta = Z/T$.
\item $G_0'$, $E_0'$, and $Z'$ to be the compromised guard-only, exit-only, and
guard-exit bandwidth, respectively;
\item $\gamma_0' = G_0'/T$; $\eta_0' = E_0'/T$; $\zeta' = Z'/T$.
\item $z = Z'/Z$; this give a parameter for compromised guard-exit bandwidth
corresponding to~$g$ and~$e$.
\end{itemize}
Some of these parameters are defined in terms of the others:
\begin{itemize}
\item $\gamma_0 = \gamma-\zeta$ and $\eta_0 = \eta-\zeta$.
\item $\zeta' = z\zeta$.
\item $\gamma_0' = g\gamma-z\zeta$ and $\eta_0' = e\eta-z\zeta$.
We can see this by noting that 
$\gamma_0'T = G_0' = gG-zZ$, from which
the expression for $\gamma_0'$ follows.  The expression for $\eta_0'$ is
derived similarly.
\end{itemize}

Thus the bandwidth contributed by guards and exits is
$G+E-Z = G_0 + E_0 + Z$.
We use these parameters to describe the probability of choosing
compromised relays as follows:
\begin{itemize}
\item A compromised guard node is chosen with probability
\[
g^* = \frac{G_0' + Z'w_{E_0}}{G_0 + Zw_{E_0}}
    = \frac{\gamma_0' + \zeta' w_{E_0}}{\gamma_0+\zeta w_{E_0}}
	= \frac{g\gamma - z\zeta(1- w_{E_0})}{\gamma-\zeta(1-w_{E_0})}.
\]
The first equality follows by dividing numerator and denominator by~$T$.
\item A compromised middle node is chosen with probability
\begin{align*}
m &= \frac{G_0'w_{G_0} + E_0'w_{E_0} + Z'w_Z}%
          {G_0w_{G_0} + E_0w_{E_0} + Zw_Z + (T - G_0+E_0+Z)} \\
  &= \frac{G_0'w_{G_0} + E_0'w_{E_0} + Z'w_Z}%
          {T - (G_0(1-w_{G_0}) + E_0(1-w_{E_0}) + Z(1-w_Z)} \\
  &= \frac{\gamma_0'w_{G_0} + \eta_0'w_{E_0} + \zeta' w_Z}%
          {1 - (\gamma_0(1-w_{G_0}) + \eta_0(1-w_{E_0}) + \zeta(1-w_Z))} \\
  &= \frac{(g\gamma-z\zeta)w_{G_0}+(e\eta-z\zeta)w_{E_0} + z\zeta w_Z}
          {1-\bigl((\gamma-\zeta)(1-w_{G_0}) + (\eta-\zeta)(1-w_{E_0}) + \zeta(1-w_Z)\bigr)}
\end{align*}
\item A compromised exit node is chosen with probability
\[
e^* = \frac{E_0' + Z'w_{G_0}}{E_0 + Zw_{G_0}}
    = \frac{\eta_0' + \zeta' w_{G_0}}{\eta_0 + \zeta w_{G_0}}
	= \frac{e\eta-z\zeta(1-w_{G_0})}{\eta-\zeta(1-w_{G_0})}.
\]
\end{itemize}

%
Taking into account the fact that the
bandwidth weighting factors
are defined in terms of the other parameters, a full specification
of the attacker for the purposes of our model consists of:
\begin{enumerate}
\item $g$, $e$, and $z$, the guard-, exit-, and guard-exit bandwidth
contributed by compromised nodes;
\item $\gamma$, $\eta$, and $\zeta$, the ratios of guard-, exit-, and
guard-exit bandwidth to total bandwidth of the network;
\item $\killprob$ and $\permitprob$, the kill- and permit- probabilities
of the attacker.
\end{enumerate}
This completes our description of the attacker.

\section{Effectiveness of the attack}
\label{sec:effectiveness}

\subsection{Theoretical results}
\label{sec:theory-eff}

We give a theoretical assessment of the effectiveness of the denial-of-service
attack in this section.  To do so, we fix the parameters describing
the attacker and consider the following experiment conducted by an 
imaginary client:
\begin{enumerate}
\item The client repeatedly chooses a path according to the path-selection 
algorithm described earlier until he choses a successful path.
\begin{itemize}
\item The path is unsuccessful if it is 
killed by the attacker (such a circuit must be compromised or controlled).
\item The path is successful otherwise.
\end{itemize}
\item If the client chooses~$K$ unsuccessful paths for some fixed~$K$, 
then he gives up.
\end{enumerate}
We say that the attacker eventually controls the client's path if
the client chooses a successful path that is controlled by the attacker.
We can then ask what the probability is that the attacker eventually
controls the client's path.

We wish to focus on the parameters~$g$, $e$, $\permitprob$, and
$\killprob$.  To do so,
we fix $\gamma=.70$, $\eta=.40$, and $\zeta=.30$,
values that we measured on the deployed Tor network in mid June~2011
(so we assume that the attacker's nodes do not have a significant
impact on the total guard or exit bandwidth fraction).

The values of $g$, $e$, and $z$ are not independent, because
compromised guard-exit bandwidth (as described by $z$) imposes a
lower bound on compromised guard bandwidth and exit bandwidth.
To achieve some desired values of $g$ and $e$, the attacker can
compromise relays with a strategy that ranges the spectrum from
compromising no guard-exit relays to compromising as many as possible.
Although the flags are assigned by Tor's directory authorities, it does
not seem difficult to configure a relay and behave in such a way that
either or both flags will be assigned.  For now, we focus on the
``prefer guard-exit'' strategy, returning to the ``avoid guard-exit''
strategy later.  To achieve
target compromise ratios $g$ and $e$, the attacker compromises
guard-exit relays until the compromised guard or exit bandwidth ratio
is either $g$ or $e$, respectively.  She then compromises/runs relays
of the other type until that desired bandwidth ratio is achieved.
In this case, $z$ is an ``observed'' value, which our analytic model
does not handle directly.  Based on simulations of this
strategy (see Section~\ref{sec:sim-results}), we have observed that this
approach typically results in $z\approx 1.5g$ 
when $g=e$ and $g\geq .05$ (the ratio is larger for smaller values
of $g$, reaching about $2.5$ for $g=.01$).  So for our analytic evaluation, 
we will take $z=1.5g$.

The appropriate value of~$K$ depends on the typical length
of time it takes for a circuit-creation attempt to fail and the time
it would normally take a client application to time out waiting for
a connection and hence give up.  
In our measurements, a failed circuit-creation attempt
either takes very little time (around $.5$ seconds) or a very long time
(around $60$ seconds).  Thus if there are attacks currently running,
we can model the attacker by taking $K$ between $2$ and $120$.
Since killing uncontrolled circuits quickly gives the attacker more
chances to control the client's circuit,
we choose~$K=120$.  Of course, this represents a \emph{very} efficient
attacker; we return to this choice at the end of
Section~\ref{sec:sim-results}.

We now determine the probability that the attacker eventually
controls the client's path.
If there are $0$ attackers in the client's guard-node list, then
the probability of eventual control is~$0$.
Since the attacker eventually controls the client's circuit if there
is some $i<K$ such that the
client chooses~$i$ unsuccessful paths followed by a path that is
controlled by the attacker,
the probability of
eventual control with $j\geq 1$ attackers in the guard-node list is
\[
\Pr[\text{even.\ ctrl., $j$ attackers}] =
\sum_{i=0}^{K-1}u_j^i\left(\permitprob\cdot\frac j3\cdot e^*\right) =
\left(\permitprob\cdot\frac j3\cdot e^*\right)\left(\frac{u_j^K-1}{u_j-1}\right)
\]
where $u_j = u_j(C)$ is the probability that the 
path~$C$ is unsuccessful when
the client has $j$ attacker nodes in his guard node list.  We compute
$u_j(C)$ as follows:
\begin{align*}
u_j(C) &= (1-\permitprob)\cdot\Pr[\text{$C$ controlled}] + \killprob\cdot\Pr[\text{$C$ compr., uncontr.}] \\
&= (1-\permitprob)\left(\frac j3 e^*\right ) +
 \killprob\cdot\left(\frac j3\Bigl(1-e^*\Bigr) + \Bigl(1-\frac j3\Bigr)e ^*+\Bigl(1-\frac j3\Bigr)(1-e^*)m\right).
\end{align*}
This is derived as follows.  $C$ is controlled if
the client chooses a compromised
guard node (probability $j/3$) and a compromised exit node
(probability~$e^*$).  $C$ is compromised and uncontrolled if either
the client chooses a compromised guard node and an uncompromised
exit node (probability $(j/3)(1-e^*)$), an uncompromised guard node and
a compromised exit node (probability $(1-j/3)(e^*)$), or uncompromised
guard and exit nodes and a compromised middle node
(probability $(1-j/3)(1-e^*)m$).

Taking into account the probability of having $j$ attackers in
the client's guard-node list,
\[
\Pr[\text{even.\ ctrl.}] =
\sum_{j=1}^3 {3\choose j}(1-g^*)^{3-j}(g^*)^j\left(\permitprob\cdot\frac j3\cdot e^*\right)\left(\frac{u_j^K-1}{u_j-1}\right).
\]

Figure~\ref{fig:ecp-ge} shows the contour plot of the eventual control
probability for the naive denial-of-service attacker
($(\permitprob,\killprob) = (1.0, 1.0)$) and the passive attacker
($(\permitprob,\killprob) = (1.0, 0.0)$) in terms of $g$ and $e$,
with the other parameters fixed as described above.  As expected,
the naive attacker controls significantly more circuits than the
passive attacker, consistently
about $2.5$ times as many, regardless of $g$ and $e$.  Perhaps something
that is not so obvious without this analysis is that for high
compromise ratios, the attacker gets more bang for her buck by
compromising additional exit bandwidth rather than guard bandwidth.
\begin{figure}
\centering
\subfloat[Naive attacker]{%
\includegraphics[]{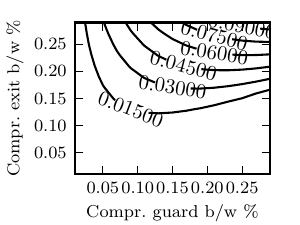}
}
\subfloat[Passive attacker]{%
\includegraphics[]{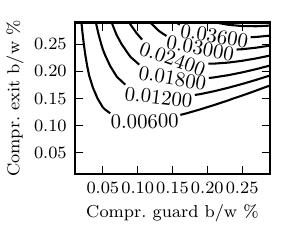}
}
\caption{\label{fig:ecp-ge}Eventual control probability for the naive
($(\permitprob,\killprob) = (1.0, 1.0)$)
and passive 
($(\permitprob,\killprob) = (1.0, 0.0)$)
attackers in terms of guard and exit node bandwidth
contributed by the attacker.}
\end{figure}

We can also vary $\killprob$ and $\permitprob$ while
keeping $g$ and $e$ constant.  Figure~\ref{fig:ecp-pq} shows the
contour plot of the eventual control probability for a 
low-resource attacker ($g=e=.01$) and a high-resource attacker
$(g=e=.10$).
\begin{figure}
\centering
\subfloat[Low-resource attacker]{%
\includegraphics[]{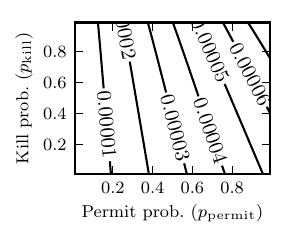}
}
\subfloat[High-resource attacker]{%
\includegraphics[]{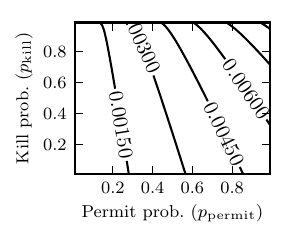}
}
\caption{\label{fig:ecp-pq}Eventual control probability for 
low resource ($g=e=.01$) and high-resource ($g=e=.1$) attackers.
}
\end{figure}
We see that the low-resource attacker with $(\permitprob, \killprob) =
(1.0, 1.0)$ eventually controls about $.007\%$ of circuits, whereas
the comparable high-resource attacker eventually controls about
$.9\%$ of circuits. Increasing compromised guard and exit
bandwidth by a factor of~$10$ increases the number of eventually-controlled
circuits by a factor of more than~$100$.\footnote{Because of the 
bandwidth capping employed by Tor during relay selection, such an
increase would probably have to be obtained by running more relays.}

So what are the resources required by our high-resource attacker?
At the time of our measurements, the guard-only, exit-only, and guard-exit
bandwidths of the deployed
network were about $1100$~MB/s, $251$~MB/s, and $751$~MB/s, respectively,
so our
attacker would have to provide about $185$~MB/s with guards
and $100$~MB/s with exits.  
An attacker following our strategy of preferring guard-exit relays
would therefore have to provide about
$100$~MB/s of guard-exit bandwidth and $85$~MB/s of guard-only bandwidth.
If she tries to keep a
low profile by running her nodes at the median bandwidth for each
type (about $435$~KB/s and $455$~KB/s for guard-exit and
guard-only, respectively), she would have to run about
$235$ guard-exits and $191$ guard-only relays.
If instead she runs nodes in the $90$-th percentile by bandwidth
(about $10$~MB/s and $9.5$~MB/s for guard-exit and guard-only, respectively),
she would need to run about
$10$ guard-exits and $8$ guard-only relays, which certainly seems
well within reason.
Our low-resource attacker really is low-resource, provided she
has sufficient bandwidth to run nodes in the $90$-th percentile:
she need only run one guard-exit and one guard-only relay.

However, these very low numbers of relays are misleading, because in practice
no relay can appear twice in a single circuit.  If the attacker
has very few nodes, then choosing paths
without replacement could have an adverse effect on the attacker.
For example, if the chosen guard makes up a significant contribution
to the attacker's guard-exit bandwidth, then the probability of
choosing a compromised exit may be much lower than that
predicted by this model.  
This is a general problem for
the attacker who has any guard-exit relays; if one of them is chosen
as a guard relay, then that decreases the attacker's available exit
bandwidth, thereby reducing the probability that the client's
circuit is eventually controlled.  The reduction in effectiveness
is significant in this analytic model.  We have performed the same
analyses with an attacker who compromises no guard-exit
relays, and she is predicted to be $5$--$8$ times more effective,
depending on the specific choices of the various parameters.
However, this is almost certainly an over-estimate, because our
analytic model assumes that there is enough bandwidth to meet the
attacker's goals.  This is not the case; for example, in the
data we use for our replay simulation in the next section, exit-only
relays with bandwidth in the $90$-th percentile and below yield
just $3.5\%$ of total exit bandwidth, so our attacker would not be
able to have $e > .035$ with this strategy.

\subsection{Simulation results}
\label{sec:sim-results}

As we have already mentioned,
our model makes a number of unrealistic assumptions.
It does not take into
account the fact that relays fail for reasons unrelated to an attacker;
for example, there may be transient network failures, a relay may
have reached its bandwidth cap, etc.
It assumes that paths are chosen with 
replacement,\footnote{To do otherwise would mean that our analytic model
would have to take into account the bandwidth contribution of a chosen relay.}
whereas Tor circuits are chosen without replacement.
And it assumes that there is sufficient bandwidth for the attacker
to compromise her desired ratios with any combination of guard-only,
exit-only, and guard-exit relays.
To assess the quality of our analytic model in light of these assumptions,
we implement a replay simulation as described below and compare the
proportion of eventually-controlled circuits predicted by the analytic
model to the proportion that are eventually controlled in the simulation.

Define a \emph{lifecycle} for a relay~$R$ to be a function
$\ell_R:\set{0,1,\dotsc}\to\set{-1, 0, 1}$.  The idea is that we
``probe'' $R$ some number of times, and $\ell_R(t)$ is the result of the
$t$-th probe.  A probe consists of constructing a circuit
of the form~$(G, R, E)$ and downloading a small file through the circuit,
where $G$ and $E$ are relays that we control.  Probe~$t$ \emph{succeeds}
($\ell_R(t) = 1$) if the file is successfully downloaded and otherwise
the probe \emph{fails}.  A probe may fail because $R$ is not in 
the consensus at time~$t$ ($\ell_R(t) = -1$) or for some other reason
such as a transient network failure, bandwidth limiting, etc.\
($\ell_R(t) = 0$).  A \emph{trial} consists of probing each
relay in the network---i.e., 
$\setst{\ell_R(t)}{\text{$R$ a relay}}$ for some fixed $t$.
We collect lifecycle data on each relay in the
deployed network by conducting some number of trials; for the
results reported here, we conducted $100$ trials over a period
of about $48$ hours.

With this lifecycle data in hand, we can simulate the denial-of-service
attack as follows:
\begin{enumerate}
\item Mark some number of the relays that are in the consensus in the first
trial as attackers; these relays are chosen according to requested
values of the parameters
$g_0$, $e_0$, and $z$ as described in our model of the attacker.%
\footnote{We ignore the distribution of IP addresses.}
Because relays have discrete bandwidths, the actual ratio of
compromised bandwidth will differ somewhat from these requested values.
We compromise relays starting at the top $90$-th percentile of bandwidth
without regard to actual reliability.
\item Attempt to build a circuit as follows:
\begin{enumerate}
\item Select $3$ guard relays from those relays~$R$ that have the
Guard flag set in trial~$0$.
\item Choose a trial~$t$ at random and try to build a circuit up
to some maximum number of times (corresponding to the parameter~$K$
in our analytic model).
\begin{enumerate}
\item Select an entry relay uniformly at random from the $3$ guard relays.
\item Select middle node and exit relays from those relays $R$
such that $\ell_R(t) \not= -1$ (exit relays must have the Exit flag
set).  Choose these two relays so that all three relays are distinct.
\item If any of the relays~$R$ has $\ell_R(t) = 0$, the build attempt fails;
try again.
\item If the circuit is compromised but not controlled, the
attacker kills it with probability $\killprob$; if it is killed, try again.
\item If the circuit is controlled, the attacker kills it
with probability $1-\permitprob$; if it is killed, try again.
\item Otherwise, the build attempt is successful.
\end{enumerate}
\end{enumerate}
\end{enumerate}
We can then analyze how many of the circuit construction attempts
result in circuits that are eventually controlled by the attacker.

We show one such comparison in Figure~\ref{fig:theory_v_sim}.  For
this analysis we consider compromise ratios $r\in[0.0, .10]$.  
We simulate the construction of many circuits with
$g = e = r$ and an attacker who prefers guard-exit relays to guard-only
or exit-only as described in the previous section.
We then set $g'$, $e'$ and $z'$ to be the actual compromised
bandwidth ratios in the simulation and compute the analytically-predicted
compromise ratio with these values (network parameters such as $\gamma$, etc.\ 
are taken to be the corresponding values in the first trial of the
replay data).
As we can see, the analytic model matches the simulation quite closely.
In particular, our unrealistic assumptions about the Tor network
do not appear to significantly impact the quality of the analytic
model of the denial-of-service attacker.
We have performed a similar analysis with the attacker who compromises
no guard-exit relays.  The simulated attacker achieves an
eventually-controlled rate of about~$.02$ for $r=.10$.  As discussed
in the previous section, this increase is not as dramatic as that
predicted by the analytic model, because our simulated attacker is
restricted by the actually-available bandwidth, and hence can
compromise at most $3.5\%$ of the total exit bandwidth.
As expected, the analytic model more consistently
over-estimates the effectiveness of the attacker.
\begin{figure}
\begin{center}
\includegraphics[]{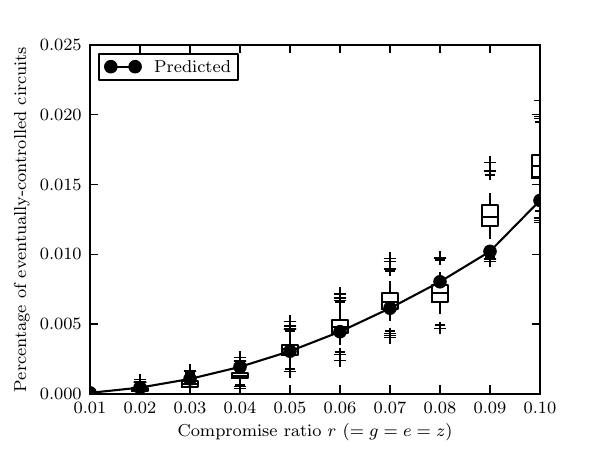}
\end{center}
\caption{Comparison of analytic model and simulation.  The boxes show the
interquartile range and median value of the percentage of eventually-controlled
circuits.  See Appendix~\ref{sec:about_data} for details.}
\label{fig:theory_v_sim}
\end{figure}

Returning to the choice of~$K$, it turns out that the value (in the
range $[2, 120]$) has little impact on the effectiveness of the
attacker as predicted by the analytic model.
The plots corresponding to
Figures~\ref{fig:ecp-ge} and~\ref{fig:ecp-pq} are almost unchanged
when we set~$K=2$;
the greatest change is in Figure~\ref{fig:ecp-ge}(a), which has the
same general contours, but starting with the lowest contour at~$.01$
and the highest at~$.06$.  We might expect that the assumption in
the model that circuits only fail because they are killed by the
attacker to have more impact with lower values of~$K$, since now
it seems much more likely that the client would give up before the
attacker could control a circuit.  This is indeed the case; with
$K=2$ the analytic model consistently over-estimates the effectiveness
of the attack as implemented in simulation.  However, the over-estimation
is by a relatively small amount.  Analyzing the simulated model with
$K=120$, we also see that almost all attempts to build a successful
circuit (controlled or not) produce one in~$\leq 15$ attempts.  
Comparing the analytic model to simulation with~$K=15$ yields a
comparison very close to that shown in Figure~\ref{fig:theory_v_sim}.

\subsection{How much is enough?}

It is natural to ask at this point whether any version of the attack
is ``effective.''  In other words:  how high must the eventual control
probability be for the attack to be considered to be a success?  We
do not give a specific answer to this question, because
it seems that it depends on the goals of the attacker, but we can
consider a couple of scenarios.

Suppose a ``script-kiddie'' just wishes to make some connections between
clients and servers, uninterested in the specific identity of either.
Then practically any eventual-control probability will do the job.  In
this case, of course, a passive attack is the route to take.

Suppose a crime-fighting unit or a repressive regime wishes to identify
some (initially unknown) users of a specific service.  At first blush,
it is not clear that denial-of-service is of any great help here,
because the users are likely to have no compromised guard nodes;
though such users will have a harder time connecting to the service,
they will be no more easily identified.  But if the goal is to make
high-profile ``examples'' of just a few users, then even a modest success
rate could be sufficient, provided it is high enough to identify users
within the regime's jurisdiction.
We address the scenario of deploying denial-of-service against a
targeted individual in Section~\ref{sec:variants}; such an attacker
is likely to have more global resources at her disposal than we are
considering here.

\section{Detecting the Attack}
\label{sec:detection-alg}

In this section we show how to detect a DoS attack as described in
the previous section.  
Briefly, the detection algorithm makes $O(n)$ probes of the network, where
a probe consists of setting up a circuit and passing data through it.
By analyzing the successful and failed probes, we can identify nodes
involved in such an attack if they exist.
We make the following assumptions
about the Tor network and the attacker:
\begin{enumerate}
\item The length of the paths used by the Tor implementation under
attack is fixed independent of (and strictly less than) $n$ and that
paths consist of distinct nodes.
\item The attacker is described by $(\killprob,\permitprob) = (1.0, 1.0)$
(the other parameters are unknown).
\item The number of compromised nodes is at least $2$ but less than $n$.  Both
bounds are reasonable, since at least two compromised nodes are required to
perform the underlying traffic confirmation attack on typical
circuits,\footnote{As shown by \citeN{syverson-ieee06},
hidden servers are vulnerable to single-node traffic confirmation attacks.}  
and an anonymity
network composed entirely of compromised nodes is of no value to an
honest user.
We address this assumption further after the proof of the theorem.%
\footnote{%
The algorithm presented
in Section~\ref{sec:implementation} does not have this restriction.}
\item The only reason
a probe fails (i.e., the circuit setup fails or the circuit dies while
data is being passed through it)
is because it is killed by an attacker on the circuit.
Of course, this ignores the fact that honest nodes may also fail,
whether due to traffic overload, intentional shutdown, etc.; we discuss
how to handle this after the proof of the theorem.
\end{enumerate}

\begin{theorem}
\label{main-theorem}
Under the above assumptions we can detect all of
the compromised nodes of the Tor network in $O(n)$ probes. 
For the case of paths of 
length~$3$ the number of probes required is at most $3n$.
\end{theorem}

\begin{proof}
Let $k$ be the length of the paths used by the Tor implementation under
consideration. 
We denote the probe consisting of the path
of length $k$ starting with $u_1$ and ending with $u_k$ with edges between
$u_i$ and $u_{i+1}$ for $i = 1, \ldots , u_{k-1}$ by $(u_1, \ldots, u_k )$.
We say a probe {\em succeeds} if the circuit is not killed, otherwise it
{\em fails}. 

Choose a set $X = \{ x_1, \ldots, x_{k-1} \}$ of $k-1$ distinct nodes, 
arbitrarily.  Perform the following set of probes:
$(x_1, y, x_2, \ldots , x_{k-1} )$ for each $y$ not in $X$.
One of three cases results.

\paragraph*{Case 1: All $n-k+1$ probes succeed}
In this case both $x_1$ and $x_{k-1}$ must be
compromised (if one is, then every probe is compromised but uncontrolled;
if neither is, then at least one probe is compromised but uncontrolled; in
either case, not all probes succeed).
For any node $y\notin X$, $y$ is compromised if and only if
the probe $(x_1, \ldots, x_{k-1}, y )$ is successful.
To test nodes in $X$, 
fix any $x\notin X$ and consider
probes of the form $(x_1,\dots,x_{i-1},x,x_{i+1},\dots,x_{k-1},x_i)$ for
each $x_i\in X$, $2\leq i<k-1$; again, $x_i$ is compromised if and only
if this probe is successful.

\paragraph*{Case 2: Among the $n-k+1$ probes, at least one succeeds and at 
least one fails}
If either endpoint were compromised, then either all probes would succeed
(if the other endpoint were compromised) or all probes would fail
(if the other endpoint were uncompromised).  Thus neither endpoint is
compromised.  But then if any of $x_2,\dots,x_{k-2}$ were compromised
every probe would fail.  Thus
in this case all of the nodes in $X$ are uncompromised, any $y$ for which
the probe failed is compromised, and any $y$ for which the probe succeeded
is uncompromised.

\paragraph*{Case 3: All $n-k+1$ probes fail}
In this case we can conclude that either
all nodes in~$X$ are uncompromised and all nodes not in~$X$ are compromised,
or at least one
of the nodes in $X$ is compromised (otherwise all nodes in~$X$ and some node
not in~$X$
are uncompromised, so at least one probe succeeds).
For each pair of nodes $x_i, x_j\in X$ consider
probes of length~$k$ of the form $(x_i,y,\dots,x_j)$, where positions~$3$
through~$k-1$ consist of $X\setminus\{x_i,x_j\}$ in an arbitrary fixed order
and $y$ ranges over nodes not in~$X$.
Suppose that for some pair $x_i,x_j\in X$ all
probes succeed.  
This second round of probes is the same as the first, but with
a different arrangement of the nodes in~$X$.  Thus the same reasoning
as in Case~1 lets us conclude that $x_i$ and $x_j$ are compromised
and we proceed as in that case to determine the status of the remaining
nodes.
Otherwise, for each pair $x_i, x_j\in X$ there is $y\notin X$
such that the probe $(x_i,y,\dots,x_j)$ fails.  In this
case, if there is at least one uncompromised node in~$X$, then
there is exactly one uncompromised node in~$X$.  Now we consider
probes of length~$k$ of the form $(x,\dots, y)$, where $x\in X$,
positions $2$ through~$k-1$ consist of $X\setminus\{x\}$
in an arbitrary fixed order, and $y$ ranges over nodes not in~$X$.
Suppose every probe of the form $(x,\dots,y)$ fails.  If there were
exactly one compromised node in~$X$, then necessarily every node not in~$X$
is uncompromised, which means that there is exactly one compromised node in the
entire network, violating our assumption that there are at least two
such nodes.\footnote{The full attack is impossible with a single 
compromised node,
though an adversary could still perform an occasional denial of
service with one such node.  A single compromised node could be
detected in a number of probes linear in $n$, though we omit the
details here.}  
Thus we conclude that if all probes $(x,\dots,y)$ fail, then
no nodes in~$X$ are compromised and all nodes
not in~$X$ are compromised.
Otherwise there are $x\in X$ and $y\notin X$ such that $(x,\dots,y)$
succeeds.  Suppose $x$ were not compromised.  Then there would be a
compromised node in $X\setminus\{x\}$ or $y$ would be compromised;
in either case the probe $(x,\dots, y)$ would fail, a contradiction.
So $x$ is compromised and hence $x$ is the only compromised node in~$X$.
Furthermore, the compromised nodes not in~$X$ are precisely those~$y$
such that the probe $(x,\dots, y)$ succeeds.

\paragraph*{Analysis}
The worst case number of probes occurs in Case~$3$ in which we do at most
$ ({k-1 \choose 2} + k - 1) (n - k+1)$ probes beyond the initial $n-k+1$
probes that define the cases.\footnote{Since some probes will be repeated,
the actual number can be made a bit smaller.}
As $k$ is assumed to be fixed independent of
$n$ this is clearly $O(n)$. For the case $k=3$ (the default for Tor), 
we notice that the initial set of probes and the first set of
probes in Case~$3$ are the same, so 
we conclude that the total number of probes is~$\leq 3n$.
\end{proof}

What happens if we apply this algorithm, but there are
no compromised nodes?  Case 1 of the proof applies, and since
every probe described in that case would succeed, we would conclude
that every relay in the network is compromised.  
In fact, the same applies if all nodes are compromised.
At this point,
presumably a human would step in to determine whether it is more likely that
no relays are compromised or the entire network is, and take
action accordingly.

A concern with this detection algorithm is that if $x_1$
is a compromised relay, then
the attacker likely notices that she is the entry guard in a
sequence of circuits in which the middle nodes traverse the
entire network.  Presumably the attacker stops killing circuits,
so we follow Case~2; we end up concluding that $x_1$ is 
uncompromised, and a further side-effect is that $x_1$ effectively
ends up framing other (uncompromised) nodes.  But how likely is
this scenario?  The probability that $x_1$ is compromised is the
fraction of compromised guard nodes (by number, not bandwidth).
Assuming some degree of human intervention, and assuming that
a relay must be identified as compromised multiple times,
the attacker escapes detection only if we repeatedly choose
her nodes in the set~$X$, which happens with low
probability.

Now we discuss how to handle the situation in which a probe may
fail for reasons unrelated to an attacker (e.g., an honest node
may fail, or there may be a transient network failure on one 
of the links).  The problem is that the detection algorithm cannot
tell what the source of the failure is.  We now define a probe to
consist of $r$ attempts to create the specified circuit, where $r$
depends on the failure rate of circuits (compromised
or not)
and the probability of error in the algorithm we find acceptable.
We report that the probe fails if all $r$ of the attempts fail, and
otherwise that it succeeds.  

We say that a probe is \emph{wrong}
if it fails but either the circuit is uncompromised or it is controlled.
Since $(\killprob,\permitprob) = (1.0, 1.0)$,
a probe consisting of $r$ independent
trials can be wrong only if (a) an honest circuit fails $r$ times in
a row or (b) a circuit with both end points compromised fails $r$ 
times in a row. 
Assume that any given circuit fails due to unreliable
nodes or edges with probability $f$.
Then, under the independence assumption, 
(a) or (b) occur 
with probability at most $f^r$, i.e., the probability that a probe
consisting of $r$ independent trials is correct is at least $1 - f^r$. 
If the algorithm performs $m$ such probes (i.e., probes $m$ circuits
overall) the probability they
are all correct is greater than $(1-f^r)^m$. Assume we require that
our algorithm correctly identifies all nodes as either honest or
compromised with probability at least $1 - \epsilon$. Then it is
easy to see (using standard approximations) that choosing 
$$r>\frac{\ln \ln (\frac{1}{1-\epsilon}) - \ln m}{\ln f}$$
is sufficient.\footnote{This bound follows from $(1-\epsilon) < (1-f^r)^m < e^{-(mf^r)}$.}

We can use our replay data to gain some insight into an appropriate
value for~$f$.  In Figure~\ref{fig:circuit-failures}, we show the
failure rate for circuit construction.  For each trial we constructed
a number of circuits by choosing three relays at random (respecting
Guard and Exit flags as appropriate) and declaring the circuit a
success if all three relays were successfully probed, and a failure
otherwise.\footnote{We have verified that this experiment predicts
circuit construction success/failure with high probability.}
We choose the relays uniformly at random (i.e., without
bandwidth-weighting), because the detection algorithm does not use
bandwidth weighting to construct its circuits.
For the purpose of choosing a lower bound
on~$r$, it suffices to find a reasonable upper bound on~$f$; from
our data, taking~$f=.45$ suffices.\footnote{In
\cite{ddos-fc09}, we indicate a failure rate of~$.2$.  In that paper
we were considering circuits as constructed by Tor using its bandwidth-weighting
algorithm.  Here we are looking at circuits constructed at random, which
are less likely to be reliable.}
If we also take
$m=7500$ (the worst-case number of probes for a $2500$-node Tor network)
and $\epsilon = .0004$ (so that we
expect less than one misidentification) we see that $r = 21$ is sufficient.
So on the deployed network, this modified algorithm would perform
$\leq 3rn = 63n$ probes.
\begin{figure}
\begin{center}
\includegraphics[]{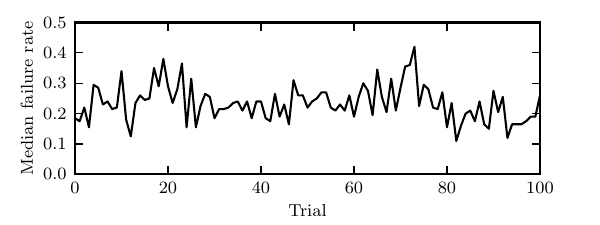}
\end{center}
\caption{Median circuit failure rates from replay data.  For
details, see Appendix~\ref{sec:about_data}.}
\label{fig:circuit-failures}
\end{figure}
Of course, we require that the above repeated attempts be independent
which is unlikely to be the case. But by spreading the repetitions out
over time we can increase our confidence that 
observed failures are not caused by randomly-occurring transient network 
failures, bandwidth limits on relays, etc.

\section{Detection in Practice}
\label{sec:implementation}

\subsection{A ``bad-relay, good-relay'' detection algorithm}
\label{sec:practical}
The detection algorithm described in Section~\ref{sec:detection-alg}
along with the measurements made above provide a reasonably practical
method for detecting the DoS attack in progress.
We can handle non-naive attackers and
reduce the number of probes of the network significantly
if we are willing to accept probabilistic detection and assume
the existence of a single honest router under our
control.  
This single honest router is a trustworthy guard
node~\cite{wright-sosp2003}.  This trust is important:  
Borisov et al.~\citeyearpar{ccs07-doa}
note that the use of (untrusted) guard nodes in general may make the
adversary more powerful when performing the predecessor
attack~\cite{wright-ndss2002}, but the assumption of a trusted guard node
avoids this problem entirely.  
By ``trusted'' we
mean that the node itself is not under the control of an attacker.
This can be arranged by installing one's own router and using
it as the guard node.
The adversary must not be able to distinguish this node from
other guard nodes on the network, for otherwise she
can choose to not attack connections from the trusted node
and remain hidden.  
Although this assumption is unrealistic with respect to a global
adversary that can observe all network traffic (because
of the very specific traffic patterns coming out of this node),
we are assuming that our adversary is local.\footnote{Tor is typically
assumed to be defenseless against a global passive adversary,
and hence such an adversary would have no need of denial-of-service
attacks.}  Furthermore, we are not arguing that every user should
have a trusted guard node, but rather just the user or organization
running the detection algorithm we describe here 
(see Section~\ref{sec:deployment} for more discussion).

The simplified detection algorithm works as follows.
First, query the Tor directory servers for a list of exit nodes,
possibly restricted by requiring some degree of stability according
to the various flags associated to each node.
Call this list of nodes
the \textit{candidate exits}.  Then, repeat the following steps $l$ times
for some value of $l$:
for each candidate node, create a circuit
where the first node is our trusted node and the second is a
candidate.  Retrieve a file through this circuit, and log the results.
Each such test either succeeds completely, or fails at some point,
either during circuit creation or other initialization, or during the
retrieval itself.  Either failure mode could be the result of a
natural failure (e.g., network outages, overloaded nodes), or an
attacker implementing the DoS attack.  
A candidate node with a high failure rate
is a \textit{suspect exit}; this failure rate can be tuned with the usual
trade-off between false positives and negatives.
Repeat an analogous process to create a list of \emph{suspect guard} nodes;
this time the circuit starts at a guard node chosen at random
and exits at our trusted node.

Once the lists of suspect guards and exits 
are generated, the following steps are
repeated $l'$ times for an appropriately chosen
$l'$.  Each possible pairing of a suspect guard and suspect exit is used
to create a circuit of length two.\footnote{An attacker controlling both
endpoints might notice that there is no middle node in such a circuit and 
kill it to defeat this detection algorithm.  This can be handled by inserting a
middle node under our control.  Alternatively, one could choose the
middle node from among the candidates not labeled as suspicious, so
as to further obscure the fingerprint of the detection algorithm.}
As above, the circuits thus
created are used to perform a retrieval, and the successes and
failures are logged.  In this set of trials, we are looking for paths
with low failure rates over the $l'$ trials.  Nodes on such paths
could be under control of the adversary, and are termed
\textit{guilty}.  

This detection algorithm performs at most $ln + l'n^2$ probes
of the network.  From the simulation results described next, we can
take $l = l' \leq 15$.  Furthermore, the number of suspects is usually
much less than $n$; we will see that it is typically about $n/10$.  Finally,
we have also determined that instead of considering every pairing of a suspect
guard and exit, for each suspect we can choose $20$ relays of
the complementary type at random and consider the $20$ corresponding pairs.
Putting all this together results in a detection algorithm that performs
$\leq 17n$ probes of the network, as compared to the $63n$ probes required
by the algorithm of the previous section.

\subsection{The algorithm in simulation}

We implement this detection algorithm against our simulation
of the denial-of-service attack described in
Section~\ref{sec:sim-results}.  Our implementation is as follows:
\begin{enumerate}
\item Mark some number of relays as attackers as described in
Section~\ref{sec:sim-results}, choosing as many guard-exit relays
as possible.\footnote{We have also analyzed this algorithm with the
attacker who compromises no guard-exit relays; the results are practically
identical.}
\item Choose a \emph{suspect cutoff rate} ($\scr$) and a \emph{guilty
cutoff rate} ($\gcr$).
\item Perform suspect-node detection:
\begin{enumerate}
\item Choose $l$ equally-spaced trials~$t_{s,0},\dotsc,t_{s,l-1}$ for
some~$l$.  
\item For each relay~$R$ with either the Guard or Exit flag set
in at least one of the $t_{s,i}$ and such that $\ell_R(t_{s,i}) \not= -1$ for
some~$i$, define
the failure rate for $R$ to be $1-m/n$, where $m$ is the number
of times $\ell_R(t_{s,i})=1$ and $n$ is the number of times
$\ell_R(t_{s,i}) \geq 0$.
Mark $R$ as a \emph{suspect} if its failure rate is~$\geq \scr$.
\end{enumerate}
\item Perform guilty-node detection:
\begin{enumerate}
\item Choose $l'$ equally-spaced trials~$t_{g,0},\dotsc,t_{g,l'-1}$ for
some~$l'$ with $t_{s,l-1}<t_{g,0}$.
\item For each pair of suspect relays $R$ and~$R'$ such that~$R$ is
a guard and $R'$ an exit in trial~$0$, define the failure rate for
the pair~$(R, R')$ to be $1-m/n$, where $m$ is the number of times
both relays are in the consensus and $(R, R')$ is successful and $n$
is the number of times both relays are in the consensus.  A pair is successful
if both relays were successfully probed and $(R, M, R')$ is not killed
by the attacker, where $M$ is a relay that we control (and hence is
always up and not an attacker).  The pair~$(R, R')$ is \emph{guilty} if
its failure rate is~$\leq \gcr$.
\item Label the relay~$R$ as \emph{guilty} if there is a guilty pair
$(R, R')$ or $(R', R)$ for some~$R'$.
\end{enumerate}
\end{enumerate}

We choose different trials for suspect and guilty node detection, because
the latter must be started after the former has completed.  
We choose a starting trial so that all suspect and guilty node detection
trials can be completed before the end of the replay data.
In our implementation, we take $t_{s,i+1} = t_{s,i}+1$, $t_{g,i+1}=t_{g,i}$,
and $t_{g,0} = t_{s,l-1}+1$ and choose $t_{s,0}$ so that
$t_{s,0}+l+l'$ is less than the total number of trials in our data.

In either phase, a
false positive is an honest relay that is labeled a suspect or
guilty, and a false negative is a compromised relay that is not so labeled.
A higher suspect cutoff rate reduces the number of relays
marked as suspects, whereas a higher guilty cutoff rate increases the
number of relays marked as guilty.  Therefore increasing~$\scr$
decreases the false-positive rate and increases the false-negative rate,
whereas increasing~$\gcr$ increases the false-positive rate and
decreases the false-negative rate.
If we assume that the attacker operates naively
(i.e., $(\killprob, \permitprob) = (1.0, 1.0)$) and that her relays are
perfect (always in the consensus and never fail), then setting
$\scr = 1.0$ will minimize the number of false positives without
admitting any false negatives.
This is because compromised relays will always fail,
whereas an innocent relay has to succeed just once to not be marked
as a suspect.  
Perfection seems unlikely, so instead we will consider an attacker
whose relays are \emph{reliable}, in that they are
simultaneously in the top $75$\% of relays
ranked by bandwidth and by number of times in the consensus (in our
simulations, the attacker
compromises reliable relays starting at the $90$-th percentile of
bandwidth).  Although this
does not seem like a strong restriction for reliability, in fact it turns
out that we have no false-negative suspects if the attacker meets this
condition even when $\scr=1.0$.
Thus the attacker must either run relays that are rarely
in the consensus (of dubious value for the attacker) or our algorithm will
label all attacking relays as suspects.  

There will still be false positives in the suspect labeling;
these are relays that are honest but unreliable, and hence have a
high ``natural'' failure rate.  These will be filtered out
during guilty-node detection, which we can see as follows.  Let $R$
be such an unreliable honest relay.  Consider any pair of suspects
$(R, R')$.  Since $R$ is unreliable, this pair will almost never
succeed, either because $R$ is out of the consensus, or $R$ is in the consensus
but fails (it does not matter whether $R'$ is honest or compromised).
Thus $(R, R')$ has a high failure rate, so is unlikely to be labeled
as guilty.  Since this is the case for every pair $(R, R')$ or $(R', R)$,
$R$ itself is unlikely to be labeled as guilty.  Thus for a perfect
attacker, setting
$\gcr = 0.0$ will ensure that we have no false negatives during guilt
detection while
minimizing the number of false positives.  It turns out that
this holds also for a merely reliable attacker, presumably because
her relays are in the consensus frequently enough that they will participate
in at least one circuit with average failure-rate~$0.0$.

It is still possible to have false-positives when detecting guilty relays.
For example, such relays could be in the consensus at least
once during suspect-detection and fail in every such trial, but then
in the consensus at least once during guilt-detection and succeed in
every such trial.  There are such relays in our replay data.
In Figure~\ref{fig:suspect_fp_perfect_borisov} we show the
false-positive rate as a function of the number of trials during
suspect and guilt detection.  As we can see, there is a slight increase in the
rate as the number of trials increases.  This is because a
false-positive is typically an unreliable relay; increasing the number
of trials during suspect detection
gives such a relay a chance to be seen by the detection algorithm,
but since it is likely to fail, it will be labeled as a suspect.
Likewise, during guilt detection, it is more likely to be seen with
more trials; if it is only in the consensus once, then it only needs to be
part of a single successful circuit to be labeled guilty, as the failure
rate of that circuit is computed with respect to the number of times
that circuit can be formed.
\begin{figure}
\begin{center}
\includegraphics[]{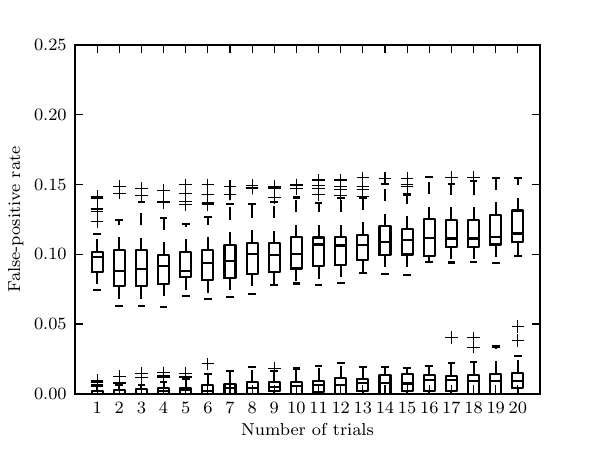}
\end{center}
\caption{False-positive rates for suspect and guilty detection with
the reliable naive attacker.  The suspect rates are in the $10$--$15$\% range
and the guilty rates are in the $0$--$1$\% range.  
For details, see Appendix~\ref{sec:about_data}.}
\label{fig:suspect_fp_perfect_borisov}
\end{figure}

Next we consider a non-naive attacker.  
As an example, consider an attacker with
$(\killprob,\permitprob) = (0.8, 0.8)$ (so, as per the results in
Section~\ref{sec:theory-eff}, compromises about $70$--$80$\% as many circuits
as the naive attacker).  Setting $\scr=1.0$ leads to an unacceptably
high false negative rate during suspect detection ($.50$--$1.0$, increasing
as we increase the number of trials), as the attacker will rarely have
a perfect success rate, even if she only has reliable relays.  
The false-positive rate is comparable to that when detecting the naive
attacker, as the attacker strategy does not affect the
behavior of non-attacking relays during suspect detection.
Figure~\ref{fig:suspect_fp_fn_reliable_tuned} shows the false-positive
and -negative rates as $\scr$ is varied (in all such figures, solid lines 
indicate false-positive rates, dashed lines false-negative rates).  
As we can see, provided we are
willing to run suspect detection for~$10$ trials, we can take
$\scr=.4$ and have no false-negatives with an acceptable false-positive rate.
\begin{figure}
\begin{center}
\includegraphics[]{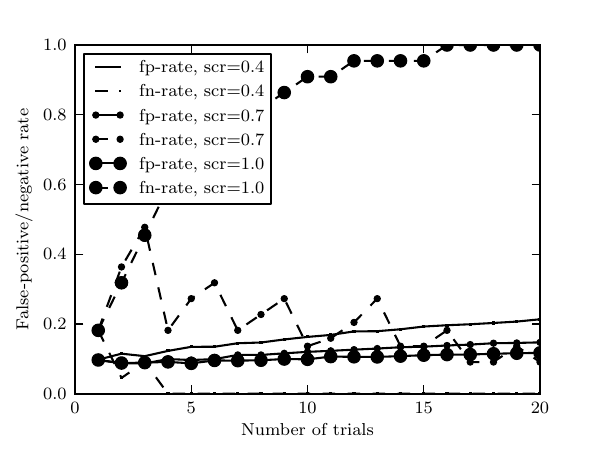}
\end{center}
\caption{False-positive and -negative rates for suspect detection with
the reliable tuned attacker with $(\killprob,\permitprob) = (.8, .8)$.
See Appendix~\ref{sec:about_data} for details, including an explanation
of the ``zig-zag pattern'' for $\scr=0.7$.}
\label{fig:suspect_fp_fn_reliable_tuned}
\end{figure}
Figure~\ref{fig:guilty_fp_fn_reliable_tuned} shows the false-positive
and -negative rates for guilty detection
as $\gcr$ is varied, keeping $\scr=.4$.  Again we see
that we can reduce the false-negative rate to~$0.0$, while maintaining
a false-positive rate of approximately~$2.5$\% by running the guilty detection
phase for $8$--$10$ trials and taking~$\gcr=.30$.\footnote{We also observe
that this value of~$\gcr$ matches nicely with the transient circuit
failure rates shown in Figure~\ref{fig:circuit-failures}; this seems to
indicate that by tuning $\gcr$, we help eliminate false-positives that
are caused by such transient failures.}
\begin{figure}
\begin{center}
\includegraphics[]{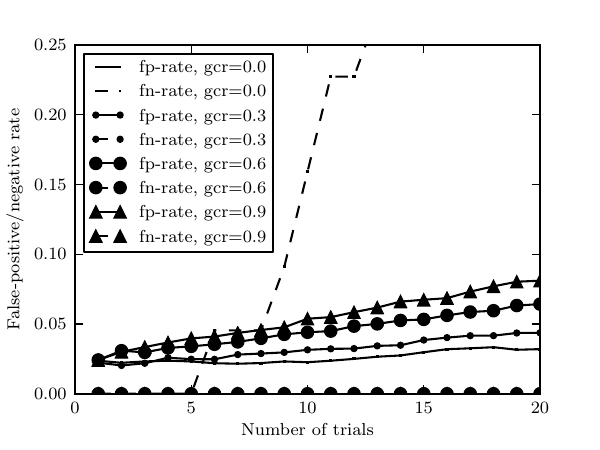}
\end{center}
\caption{False-positive and -negative rates for guilty detection with
the reliable tuned attacker with $(\killprob,\permitprob) = (.8, .8)$.
See Appendix~\ref{sec:about_data} for details.}
\label{fig:guilty_fp_fn_reliable_tuned}
\end{figure}

Finally we consider how well our detection algorithm works as the
attacker reduces her kill probability from $1.0$ (the naive attacker)
to $0.0$ (the passive attacker), keeping her permit probability
at $1.0$.  Of course, if $\killprob=0.0$, then the attacker cannot
be detected; our interest here is how quickly our algorithm loses
effectiveness.
Figure~\ref{fig:suspect_fp_fn_rates} shows the suspect
false-positive and -negative rates for various values of~$\scr$; here
we have run suspect detection for~$15$ trials.
\begin{figure}
\begin{center}
\includegraphics[]{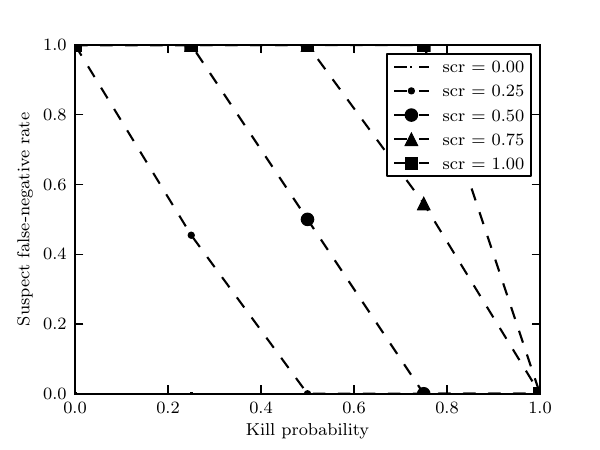}
\end{center}
\caption{False-negative rates for suspect detection with
the reliable attacker with varying kill probabilities and $\scr$ values.
See Appendix~\ref{sec:about_data} for details.}
\label{fig:suspect_fp_fn_rates}
\end{figure}
As we can see, even if the attacker lowers her kill probability 
to~$.5$ in order to escape detection, we will still have a nearly $0$\%
false-negative rate during suspect detection, provided we lower $\scr$
to~$.25$.  Of course, lowering $\scr$ increases the
false-positive rate; in this case, lowering to $.25$ from $1.0$
increases the false-positive rate from about $10$\% to about $25$\%
(this is independent of the attacker's kill probability, since this does
not affect the behavior of non-attackers during suspect detection).
This rather high false-positive rate is only an issue if it
persists through guilt detection.  In 
Figure~\ref{fig:practical_guilty_comparison}
we show the false-positive and -negative rates for guilt detection as $\gcr$
and the kill probability are varied, where we fix
$\scr$ at $.25$ and run both suspect and guilt detection for~$15$ trials.
\begin{figure}
\begin{center}
\includegraphics[]{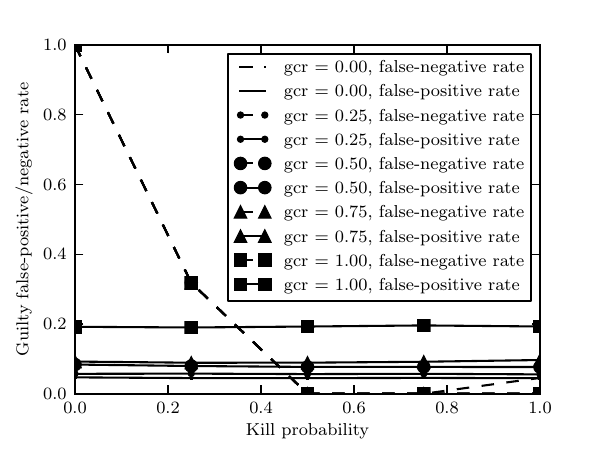}
\end{center}
\caption{False-negative rates for guilty detection with
the reliable attacker with varying kill probabilities and $\gcr$ values.
See Appendix~\ref{sec:about_data} for details.}
\label{fig:practical_guilty_comparison}
\end{figure}
Clearly, just about any value of $\gcr > 0.0$ suffices to reduce the
false-negative rate to essentially $0$\%, even when the attacker's
kill-probability is~$.5$.  And provided $\gcr < 1.0$, the false-positive
rate is about~$5$\%.  
This seems like a reasonable compromise if the
primary goal is to identify compromised relays.  

\subsection{How good is good enough?}

Just as we can ask how effective the denial-of-service attack must be,
we can also ask how effective any detection algorithm must be.  For example,
is a false-positive rate of $5$\% acceptable?  Again, we do not answer
this directly, because this seems to be more of a matter for policy
(of course, if the false-positive rate were $90$\%, the policy would
be easy to settle).  We assume that any automated algorithm would really
flag ``guilty'' relays as being relays that deserve further inspection.
Probably such inspection would be carried out by humans.  The goal of
algorithms such as those presented here is to reduce the workload of
humans to a manageable level by clearing many relays of suspicion
automatically.

\section{Variants of the attack}
\label{sec:variants}

The attacker we have described in Section~\ref{sec:borisov-attack},
and on which our detection algorithms are based, kills circuits
unconditionally according to the parameters $\killprob$ and $\permitprob$.
However, an attacker may be interested in a \emph{contextual} attack,
for example only attacking connections to particular hosts or traffic
of a certain type.  Our analytic model and detection algorithms
handle contextual attacks more-or-less well depending on the
specifics of the context.

On one end of the spectrum are contexts in which circuit membership
can be determined by a relay in any position on the circuit.
An example such context is bulk-download traffic.  In this case,
$\killprob$ and $\permitprob$ are the probabilities of killing and
permitting circuits that satisfy the context, respectively.
The analytic model is unchanged.  The only change to the detection
algorithm is what constitutes a ``probe'' of a relay; now it is a
circuit that satisfies the context.

Somewhat more challenging are contexts in which circuit membership 
can only be determined by relays in certain positions.  An example
is circuits that connect to a specific host; only the exit relay
can determine whether the context is satisfied or not.
This means that if a guard or middle
node determines that it is the only attacker node in the circuit,
it does not have enough information to determine whether the circuit
satisfies the attack context.  Our analytic model can be adapted to
handle this attacker by adjusting the calculation of $u_j(C)$, the
probability that the attempt to build circuit~$C$
is unsuccessful when the client
has $j$~attacker nodes in his guard node list.  Recall that
\[
u_j(C) = (1-\permitprob)\cdot\Pr[\text{$C$ controlled}] + \killprob\cdot\Pr[\text{$C$ compr., uncontr.}].
\]
What we need to do is to define two kill probabilities:
$p_{\mathrm{kill, aware}}$ and $p_{\mathrm{kill, unaware}}$.
The former is the kill probability for relays that can determine whether or not
the context is satisfied; the latter is for relays that cannot.  We can then
rewrite the term $\killprob\cdot\Pr[\text{$C$ compr., uncontr.}]$ to
take into account both kinds of relays.  For the example at hand,
this term would be
\[
p_{\mathrm{kill,aware}}\Bigl(1-\frac j3\Bigr)e^* +
p_{\mathrm{kill,unaware}}\left(\frac j3\Bigl(1-e^*\Bigr) + \Bigl(1-\frac j3\Bigr)(1-e^*)m\right).
\]
Note that we only need one value for $\permitprob$, because if the
relays can communicate enough to determine that they are all on the
same circuit, then presumably any context-aware relay can convey the
context information to the other relays as well.  How the
``bad-relay, good-relay'' detection algorithm fares depends on the
relation between $p_{\mathrm{kill,aware}}$ and
$p_{\mathrm{kill,unaware}}$.  If they are equal, then no change is needed
(this is unsurprising, since in this case, the analytic model is also
unchanged).  But suppose that $p_{\mathrm{kill,unaware}} = 0.0$---i.e.,
an uncontrolled
circuit is only killed by an exit relay that observes a connection
to the desired host.  As described, the first phase of our detection
algorithm will have an unacceptable false-negative rate for
guards, as they will never kill circuits.
It is possible to adapt the algorithm so that \emph{all} guards
are initially labeled as suspects.  This increases
the cost of the second phase.\footnote{And we would no longer have the
option of only comparing each suspect to a fixed number of other
suspects for guilt detection as described at the end of
Section~\ref{sec:practical}, since that strategy appears to rely on
the assumption that suspect detection does not have a very high
false-positive rate.}
However, our preliminary experiments indicate
that the false-positive and -negative rates for guilt detection are
essentially unchanged from those shown in 
Figure~\ref{fig:guilty_fp_fn_reliable_tuned}.
And there is a trade-off for the attacker
here.  By setting $p_{\mathrm{kill,unaware}}=0.0$, her guards are
not suspiciously killing circuits which, in all likelihood, are not
even connecting to the targeted endpoint, and this makes our detection
algorithm more time-consuming to run.
On the other hand, her
attack is less effective:
our analytic model predicts that 
she eventually controls about half as many circuits as
for the context-independent attack.

An attack targeted at an individual user is much more difficult
for our model and detection algorithms.  It is unlikely that such
an attack would be launched using only relays; it seems much more likely
that the attacker controls the user's ISP and performs denial-of-service
in order to control the exit relay.  There is no need to control the
entry node in this case, as the ISP can view the traffic between 
the client and entry relay, and that is sufficient.  Obviously a centralized
authority running our detection algorithm would not see the attack,
since it would not be attacked at all.  And the ISP could easily see
if the individual user were running the algorithm and react accordingly.
This highlights the restrictions of the locality assumption of
Section~\ref{sec:borisov-attack}:  we consider denial-of-service attacks
that are run from individual Tor relays, not from more powerful attackers
who may have more global resources (such as an ISP).  We also note that
such an attacker might have sufficiently global resources as to be
able to launch a purely passive (and hence undetectable) attack using
timing analysis techniques such as those described
by~\citeN{murdoch-zielinski:pet2007}.

In order to reduce the effectiveness of the detection algorithms,
an attacker may ``frame'' honest nodes, under the reasonable assumption
that the information content provided by 
a detection algorithm that produces too many false positives
would be too low to be useful.  
One approach is to kill circuits in an attempt to frame the honest
nodes that are on those circuits, but
it is not clear that such framing
would be effective with the ``bad-relay, good-relay'' algorithm.
In the first phase (suspect detection), only one ``wild'' node is probed
at a time; thus no framing is possible in this phase, and the only
false positives are unreliable relays.  In the second phase (guilt detection),
relays are paired up, and so an attacker might try to frame honest nodes.
However, in this phase, a relay is labeled as guilty if it is
``too reliable;'' since the only honest nodes to make it into this phase
are unreliable, and nothing an attacker can do will make them more
reliable, it seems difficult to frame any nodes during this phase either.
Another approach is to employ denial-of-service attacks to flood
an honest relay with traffic, thereby making it a suspect in the first
phase of the detection, then stopping the denial-of-service in order
to make it guilty in the second.  This would require the attacker to
know when the detection algorithm is being run, and it is unclear
whether one can reasonably defend against an attacker with this level
of knowledge.

\section{Running the detection algorithms}
\label{sec:deployment}

The algorithms we have proposed here are not intended to be run
by individual Tor users (in contrast to the algorithm described by
\citeN{das-borisov:securing-tor-tunnels}).  Thus there is legitimate
concern as to how these algorithms can be employed in a way that does
not fundamentally alter the decentralized nature of Tor.  This
decentralization is important.  If it is known that the detection
algorithms are run by a small, specific set of relays, then the
attacker can easily avoid detection by simply permitting connections
to/from those relays.  It would be very interesting to see if a 
distributed version of these algorithms could be implemented.
One possible version would have a large percentage of relays
involved in running the detection algorithm, with each relay testing
a small portion of the network.
Relays would then vote on the results in a manner analogous to 
how the directory authorities currently vote on flags, etc.

A certain amount of centralization is almost certainly lost though
in the final stages.  It is problematic at best to allow a fully automated
process to block operators from participating in Tor.  If a group of
relays is deemed by these algorithms to be launching DoS attacks, it
seems almost necessary that humans ultimately step in to determine an
appropriate course of action.

\section{Related work}
\label{sec:related}

The arms race between attackers and defenders in anonymity systems has
a long history. System designers aim to prevent attacks, or failing
that, to detect and respond to them. In turn, attackers attempt to
evade or bypass prevention and detection mechanisms. Here, we briefly
survey some related work in this arms race.

The MorphMix system~\cite{morphmix:wpes2002}, like Tor, is a
peer-to-peer system for low-latency anonymous communication on the
Internet. The system's design includes a collusion detection
mechanism. 
Later, \citeN{morphmix:pet2006} showed
that local knowledge of the network does not suffice to detect
colluding adversaries. 

\citeN{danezis:wpes2003} propose a detection algorithm for active
attacks in mixes, based upon self-addressed heartbeat messages sent
through the mix itself. This algorithm is concerned with an
$(n-1)$ attack, where an attacker floods an honest node with fake
messages to enable the linking of the sender and receiver of a single
message; a heartbeat is used to attempt detection of such attacks. The
heartbeat mechanism has some parallels to our probing mechanisms,
though the attacker models are quite different.

Murdoch~\citeyearpar{HotOrNot,steven-thesis} examines the use and
detection of various covert channels in attacks on anonymity
systems. The types of attack algorithms and corresponding detection
mechanisms again illustrate the arms race, though they do not map to
the attacker model we examine.

\citeN{das-borisov:securing-tor-tunnels}  propose a detection
algorithm intended to be used by individual Tor users in order to
avoid circuits compromised by a DoS attacker; 
this work is the closest to ours.  
The algorithm itself is similar to
our ``bad relay-good relay'' detection algorithm 
(though their algorithm might be described as ``good relay-bad relay'').
Their goal is to allow clients to
identify (potentially) compromised circuits over a short timeframe in order
to avoid using them, rather than to identify specific nodes implementing
a DoS attack over a longer timeframe.  
Their approach can be used
to mitigate the risk to a user from an ISP-level attacker which, as
we discuss in Section~\ref{sec:variants}, our algorithm cannot do
(although as we also note there, such an attacker might very well
gain enough information from a passive attack).
The cost of running their algorithm (in terms of number of circuits created) 
appears lower than that of our algorithms.
However, it is not clear that the overhead on the entire network would
actually be lower if all clients were to implement their algorithm.

\section{Conclusion}
\label{sec:conc}

The denial of service attack on Tor-like networks is potentially quite
powerful, allowing an adversary to break the anonymity of
users at a rate much higher than when passively listening.
We have provided a careful analysis of the parameters that define such
an attack, as well as an analytic model of the attacker's effectiveness.
We have tested this model against a simulation based on replaying
data collected from the deployed Tor network and seen that it is
accurate.  We have also shown that the power of the denial-of-service
attack comes at a price by giving two algorithms that detect any
such attacker by constructing a number of circuits that is linear
in the number of relays in the network.  One such algorithm is
deterministic and proved correct given a set of assumptions about the
network and attacker, and the other probabilistic and shown to be
effective using our replay simulation technique.
We finish by discussing how these algorithms fare in the face of
attackers who deviate from our model.

\appendix

\section{About the data}
\label{sec:about_data}

All of the data described 
in this paper, as well as the programs used to collect
and analyze it, are publicly available at Wesleyan University's
WesScholar site at
\url{http://wesscholar.wesleyan.edu/compfacpub}, in the section for
this paper.
When specific numbers
are indicated, they refer to data collected approximately $10$--$11$~June~2011
(timestamps are included in the data).
This dataset consists of $100$ trials.
Following are some notes on
how specific figures were produced:

\paragraph*{Figure~\ref{fig:theory_v_sim} (Comparison of theoretical model
to simulation)}
For each $r\in\set{.01,.02,\dots,.10}$, relays in the replay data were
compromised according to the algorithm described in
Section~\ref{sec:theory-eff} to reach the target goal of $g = e = r$.
Then $10,000$ circuits were constructed and $1,000$ bootstraps are
performed.  Each bootstrap consists of selecting $10,000$ circuits
from the population, sampling with replacement, and recording the 
percentage of selected circuits that are controlled by the attacker.
The median and interquartile range of the $1,000$ bootstraps is shown.
Then the analytically-predicted value is computed, using the actual
guard, exit, and guard-exit compromise ratios and the actual values of
$\gamma$, $\eta$, and $\zeta$ as in the collected data.

\paragraph*{Figure~\ref{fig:circuit-failures} (Circuit construction
failure rate)}
For each trial in
the replay data, we constructed~$100$ circuits and noted whether each
was successful or not.  We then sampled these circuits with
replacement $100$ times and noted the proportion of failed circuits.
We do the sampling~$10$ times per trial and show the median
failure rate.

\paragraph*{Figure~\ref{fig:suspect_fp_perfect_borisov} (Suspect and
guilty false-positive rates for the reliable naive attacker)}
For each number of trials~$n$, we run the suspect detection phase $100$
times.  Each time we choose a starting trial at random, run the algorithm
over~$n$ trials, and record the false-positive rate.  We display the median
rate and the inter-quartile ranges.  Taking $5$ trials as our best
number of trials for suspect detection, for each~$n$ we run the suspect
and guilty detection phases together $100$ times.  Each time we choose a
starting trial at random, run the suspect phase for $5$ trials, the guilty
phase for $n$~trials, and record the false-positive rate for guilt detection.
We display the median rate and inter-quartile ranges.

\paragraph*{Figure~\ref{fig:suspect_fp_fn_reliable_tuned} (False-positive
and -negative suspect rates for reliable tuned attacker)}
For each number of trials~$n$ and value of $\scr$, 
we run the suspect detection phase for $n$
trials with the given value of $\scr$.
Then we compute the false-positive and \mbox{-negative}
rates for suspect detection.  This is repeated~$100$ times, each time choosing
a starting trial at random.  We plot the median false-positive and 
false-negative rate for each combination of $n$ and $\scr$.
The ``zig-zag'' pattern for the false-negative rate
when $\scr=0.7$ is an artifact of how
the number of trials that a guilty relay must pass to avoid
detection changes as the total number of trials increases.  This number
is $1$ for $1$--$3$ trials; $2$ for $4$--$6$ trials; $3$ for
$7$--$9$ trials; etc.  If the number of trials to pass to avoid 
detection does not increase, then the false-negative rate will increase.
As $\scr$ decreases, the number of trials to pass jumps less frequently,
and the false-negative rate is already relatively low, so the pattern is
not as obvious at lower values.

\paragraph*{Figure~\ref{fig:guilty_fp_fn_reliable_tuned} (False-positive
and -negative guilty rates for reliable tuned attacker)}
For each number of trials~$n$ and value of $\gcr$, 
we run the suspect detection phase for $10$
trials with $\scr=.4$, followed by guilty detection for $n$ trials with the
given value of~$\gcr$.  Then we compute the false-positive and -negative
rates for guilty detection.  This is repeated~$100$ times, each time choosing
a starting trial at random.  We plot the median false-positive and 
false-negative rate for each combination of $n$ and $\gcr$.

\paragraph*{Figure~\ref{fig:suspect_fp_fn_rates} (False-negative suspect
rates for varying kill probability)}
For each value of $\scr$ and $\killprob$, suspect detection is run for $15$
trials and the false-negative rate is recorded.  This is repeated $100$ times,
each time choosing a starting trial at random.  The median false-negative
rate is plotted for each value.

\paragraph*{Figure~\ref{fig:practical_guilty_comparison} 
(False-positive/negative rates for varying kill probability)}
For each value of $\gcr$ and $\killprob$, suspect detection is run for $15$
trials with $\scr$ fixed at $.25$.  Then guilt detection is run for $15$
trials with the given value of $\gcr$ and the false-positive and -negative
rates are recorded.  This is repeated $100$ times, each time choosing a
starting trial for suspect detection at random.  The median rate is plotted
for each pair of values.

\section*{Acknowledgment}
We would like to thank George Bissias
for noting an error in an early version of the proof of 
Theorem~\ref{main-theorem}. 
We also thank the many anonymous reviewers who provided helpful
comments and suggestions for improving this paper, particularly the
suggestions that led to the content of Section~\ref{sec:variants}.

\bibliographystyle{plainnat}
\bibliography{dos-in-tor}

\end{document}